\documentclass{CSML}

\pdfoutput=1

\usepackage{lastpage}

\def\dOi{13(4:30)2017}
\lmcsheading%
{\dOi}
{1--\pageref{LastPage}}
{}
{}
{Aug.~31, 2017}
{Dec.~28, 2017}
{}

\usepackage[hidelinks]{hyperref}

\theoremstyle{plain}

\newcommand{\XC}{X_c}
\newcommand{\SX}{\widehat{X_c}}
\newcommand{\SK}{\widehat{K}}

\newcommand{\Ix}{Ix}

\newcommand{\NN}{\text{\rm \hbox{I\kern-.2em\hbox{N}}}}
\renewcommand{\O}{{\cal O}^e} 

\newcommand{\ee}{{effectively enumerable }}

\renewcommand{\O}{{\mathcal O}^e}

\begin{document}

\title[ The Rice-Shapiro Theorem in Computable Topology]{ The Rice-Shapiro Theorem in Computable Topology}

\author[M.~Korovina ]{Margarita Korovina\rsuper a}	
\address{{\lsuper a}A.P. Ershov Institute of Informatics Systems, SbRAS, Novosibirsk}	
\email{rita.korovina@gmail.com}  
\thanks{{\lsuper a}The research leading to these results has received funding
 from the DFG grants WERA MU 1801/5-1 and  CAVER BE 1267/14-1, the RFBR grant 17-01-00247}	

\author[O.~Kudinov]{Oleg Kudinov\rsuper b}	
\address{{\lsuper b}Sobolev Institute of Mathematics, SbRAS, Novosibirsk}	
\email{kud@math.nsc.ru}  





\begin{abstract}
 We provide requirements on   \ee $T_0$--spaces which guarantee  that  the Rice-Shapiro theorem holds for  the computable elements of these spaces.
  We show that  the relaxation of these  requirements leads to the classes of \ee $T_0$--spaces where  the  Rice-Shapiro theorem does not hold. We propose two constructions that generate \ee $T_0$--spaces with particular properties from $wn$-families and computable trees without computable infinite paths. Using them we propose examples that give a flavor of this  class.
\end{abstract}

\maketitle

\section*{Introduction}\label{sec_intr}
The paper is a  part of the ongoing research program \cite{Korovina_mscs,Korovina_CiE15,Korovina_PSI14}  that follows the tradition started in \cite{Cenzer_98,Cenzer_99,Brodhead_08,Spreen98,Spreen95,Spreen96,Weihrauch_cpo} to merge classical recursion theory and computable analysis in order to study
 computable (effective) topology.

  One of the natural ideas in computable topology is an  adaptation of the results from classical computability (recursion) theory concerned with  the lattice of computably enumerable sets to computable analysis for  studying the structures of computable elements and effectively open sets  of topological  spaces.
  The major obstacle for the achievement of this goal is the different nature of discrete and continuous data. An approach to overcome these difficulties
   is to figure out which particular classes of topological spaces fit better for generalisations of classical results.
 Following this direction in this paper we consider the famous Rice-Shapiro theorem.
 In classical computability theory the theorem  provides a~clear charcterisation of computably enumerable index sets that
 leads to a~simple description of effectively enumerable properties of program languages in computer science.
  Historically,  generalisations of the Rice-Shapiro theorem have been first proven for  the algebraic domains  \cite{Ershov_1}, for the weakly effective $\omega$-continuous domains \cite{Spreen84} and later on for  the effectively pointed topological spaces that expand the weakly effective $\omega$-continuous  domains \cite{Spreen96}.

  In this paper as  promising candidates for the realisation of this program we  consider \ee $T_0$-spaces with  conditions on the family of basic neighborhoods of computable elements that guarantee the existence of a~prin\-cipal computable numbering. In \cite{Korovina_CiE15} we have  already shown that for this class of topological spaces generalisations of Rice's theorem hold.
  However, as we demonstrate in this paper, these conditions are not sufficient to establish the Rice-Shapiro theorem.
   Towards the  Rice-Shapiro theorem we enhance the requirements on   \ee $T_0$--spaces that result in  a new class called the class of modular $T_0$--spaces. It turns out that, on  the one hand, the Rice-Shapiro theorem holds for this class, on the other hand, this class extends the weakly effective $\omega$-continuous domains.
   The paper is organised  as follows.
  In Section~\ref{sec_preliminaries} we give basic notions and definitions. In Section~\ref{sec_comp_element} we recall the notion of a computable element and
   conditions on the family of basic neighborhoods of computable elements that guarantee the existence of a~prin\-cipal computable numbering.  Section~\ref{sec_Rice_Shapiro} contains the definition of  the class of  modular $T_0$-spaces and the proof of  the Rice-Shapiro theorem for this class. In Section~\ref{sec_constructions} we show a few general approaches for constructing \ee $T_0$-spaces with particular properties.  In Section~\ref{sec_Counterexamples} we use them to produce counterexamples.

\section{Preliminaries}\label{sec_preliminaries}
\subsection{Recursion theory}
We refer the reader to \cite{Rogers} and  \cite{Soarehbook} for  basic definitions and fundamental concepts of recursion theory
 and to \cite{Ershov_Num_73_1,Ershov_Num_75_2,Ershov_Num_2} for  basic definitions and fundamental concepts of numbering theory.
 We recall that, in particular, $\varphi_e$ denotes the partial  computable (recursive) function with  an index $e$ in the Kleene numbering.
 Let $\varphi^s_e ( x ) = \varphi_e ( x )$
 if computation requires not more than $s$ steps and $x \leq s$, otherwise $\varphi^s_e (x) = -1$. In this paper we also use the notations $W_e=\rm {dom}(\varphi_e)$, $W^s_e=\{x\mid \varphi^s_e (x) \geq 0\}$, and $\pi_e=\rm {im}(\varphi_e)$.
A~sequence $\{V_i\}_{i\in\omega}$ of computably enumerable   (c.e.) sets is {\em computable} if
 $\{(n,i)|\, n\in V_i\}$ is computably enumerable.
 It is worth noting that this is equivalent to
 the existence of a computable function  $f:\omega\to\omega$ such that $V_i=W_{f(i)}$.  We denote the canonical computable sequence of all finite sets as $\{D_n\}_{n\in\omega}$. A~sequence $\{V_i\}_{i\in\omega}$ of finite sets is {\em strongly computable} if there exists a computable function $h:\omega\to\omega$ such that $V_i=D_{h(i)}$ for  all $i\in\omega$.
 A~sequence $\{V_i\}_{i\in\omega}$ of finite sets is a  {\em presentation} of $V=W_e$ if $V_i=W^i_e$ for all $i\in\omega$.  It is clear that the sequence $\{V_i\}_{i\in\omega}$ is strongly computable, $V_s\subseteq V_{s+1}$ and $V=\bigcup_{i\in\omega} V_i$.
A numbering of a set $Y$ is a surjective map $\gamma:\omega\to Y$.
\subsection{Weakly effective $\omega$-continuous domains}
In this section we present some background on domain theory. The reader can find more
details  in \cite{Abramsky_94,SHLG94,Gierz}.
Let $D = (D;\bot,\leq)$ be a partial order with a least element $\bot$.
 A subset $A\subseteq D$ is \emph{directed}
if $A\neq \emptyset$ and $(\forall x,\,y \in A)(\exists z\in A)(x\leq z \wedge x\leq z)$. We say that $D$ is a \emph{directed complete partial order}, denoted dcpo, if any directed set  $A\subseteq D$ has a supremum in $D$, denoted $\bigsqcup A$. For two elements $x,\, y\in D$
we say $x$ is \emph{way-below} $y$,
denoted $x\ll y$, if whenever $y\leq \bigsqcup A$ for a directed set $A$, then there exists $a\in A$ such that $x\leq a$.
We say that $B\subseteq D$ is a \emph{basis} (base) for $D$ if for every $x\in D$
the set $approx_{B}(x)=\{y\in B| y\ll x\}$ is directed and $x=\bigsqcup approx_{B}(x)$.
We say that $D$ is \emph{continuous} if it has a basis; it is \emph{$\omega$--continuous} if it has countable basis.

\begin{defi}\cite{Gierz}
 Let $(D;B,\beta,\leq,\bot) $ be an $\omega$--continuous domain where $B$ is a basis, $\beta:\omega\to B$ is a numbering of the basis. We say that $D$ is a weakly effective if  the relation $\beta(i)\ll \beta(j)$ is computably enumerable.
\end{defi}
\begin{prop}[Interpolation Property]\cite{Abramsky_94,Gierz}\label{interpolation_prop}
Let $D$ be a continuous domain and let $M\subseteq D$ be a finite set that $(\forall a\in M)\, a\ll y$. Then there exists $x\in D $such that $M \ll x\ll y$
holds. If $B$ is a basis for $D$ then $x$ may be chosen from $B$.
\end{prop}
\subsection{Effectively enumerable $T_0$-spaces }In the  paper we work  with  the   \ee $T_0$-spaces. The class of \ee topological spaces has been   proposed in \cite{Korovina_CCA08}. This is a wide class  containing weakly effective $\omega$--continuous domains, computable metric spaces and positive predicate structures \cite{Korovina_pps} that retains  certain natural effectivity requirements which allow us  to represent important concepts of effective topology.

Let $\left( X,\tau,\alpha \right )$ be a topological space, where $X$
is a non-empty set, $B \subseteq 2^{X}$ is
 a base of the topology $\tau$  and $\alpha:\omega\to B$ is a numbering.

\begin{defi}\cite{Korovina_CCA08}\label{ee}
A topological space $\left( X,\tau,\alpha \right )$ is {\em effectively
enumerable} if the following conditions hold.
\begin{enumerate}
\item[(1)] There exists a computable function $g:\omega\times\omega\times\omega\to \omega$ such that
\begin{eqnarray*}
\alpha( i)\cap\alpha( j)=\bigcup_{n\in\omega}\alpha( g(i,j,n)).
\end{eqnarray*}
\item[(2)] The set $\{ i\mid \alpha( i)\neq \emptyset\} $ is computably enumerable.
\end{enumerate}
\end{defi}

In the following we assume that an effectively enumerable topological space $\left( X,\tau,\alpha \right )$ is given.
Further on we will often abbreviate $\left( X,\tau,\alpha \right )$ by $X$
if $\tau$ and  $\alpha$ are clear from a context. We use the following notions of
an effectively open set and a computable sequence of effectively open sets.
\begin{defi}\cite{Korovina_CCA08}\label{def_eo}
\begin{enumerate}
 \item A set
 ${\mathcal O}\subseteq X$ is {\em effectively open} if
there exists a computably enumerable set $V$  such that
\begin{eqnarray*}
{\mathcal O}=\bigcup_{n\in V}\alpha(n).
\end{eqnarray*}
\item A sequence $\{{\mathcal O}_n\}_{n\in\omega}$ of effectively open sets is called {\em computable} if there exists a computable  sequence  $\{V_n\}_{n\in\omega}$ of computably enumerable sets such that
   ${\mathcal O}_n=\bigcup_{k\in V_n}\alpha(k)$.
\end{enumerate}
\end{defi}

\noindent Let ${\mathcal O}_X$ denote the set of all open subsets of $X$ and $\O_X$ denote the set of all effectively open subsets of $X$.
\begin{defi}\cite{Korovina_PSI14}
\hfill  
\begin{enumerate}
\item A numbering $\beta:\omega\to\O_X$ is called {\em computable} if   $\{\beta(n)\}_{n\in\omega}$ is   a computable sequence.
\item A numbering $\beta:\omega\to\O_X$ is called  {\em  principal computable} if
it is computable and every computable numbering  $\xi$ is computably reducible to $\beta$, i.e.,
there exists a computable function $f:\omega\to\omega$ such that $\xi(i)=\beta(f(i))$.
\end{enumerate}
\end{defi}
\begin{prop}\cite{Korovina_PSI14}\label{pcnumb_eos}
There exists a principal computable numbering $\alpha^{e}_X$ of
 $\O_X$.
\end{prop}
\subsection{Computable Elements}\label{sec_comp_element}
In this section we work with  \ee $T_0$-spaces  $\left( X,\tau,\alpha \right )$
and use the following notion of a computable element.
\begin{defi}\cite{Korovina_CiE15}\label{comp_element}
 \begin{enumerate}
\item  An element $x\in X$ is called {\em computable}  if the set $A_x=\{n|x\in\alpha(n)\}$ is computably enumerable.
    \item A sequence $\{a_n\}_{n\in\omega}$ of computable elements is called {\em computable} if the sequence  $\{A_{a_n}\}_{n\in\omega}$ of computably enumerable sets is computable.
    \end{enumerate}
\end{defi}

\noindent It is easy to see that the definition above generalises  the notions of a~computable real number, a~computable element of a computable metric space \cite{Brattka_baire},
 a~computable element of a~weakly effective $\omega$-continuous domain \cite{Weihrauch_cpo,Spreen84}
  and agrees with the notion of a~computable element of a computable topological space \cite{Grubba09}.
 It is worth noting that there are \ee topological spaces without computable elements \cite{Korovina_mscs}.
 Further on we use the following notations.
 \begin{itemize}
 \item $\XC$ denotes the set of all computable elements of  $X$.
 \item $A_x=\{n|x\in\alpha(n)\}$.
 \item For $K\subseteq \XC$, $\SK=\{A_a\mid a\in K\}$, in particular $\SX=\{A_a\mid a\in \XC\}$.

 \end{itemize}

\begin{defi}\label{numb_c_elem}\cite{Korovina_CiE15}
\hfill  
\begin{enumerate}
\item A numbering $\gamma:\omega\to\XC$ is called {\em computable} if   $\{A_{\gamma(n)}\}_{n\in\omega}$ is a  computable sequence.
\item A numbering $\gamma:\omega\to\XC$   is called {\em  principal computable} if
it is computable and every computable numbering  $\xi$ is computably reducible to $\alpha$, i.e.,
there exists a computable function $f:\omega\to\omega$ such that $\xi(i)=\gamma(f(i))$.
\end{enumerate}
\end{defi}

\noindent Now we address the natural question  whether for an \ee $T_0$-space there exists  a computable numbering of the computable elements.
 First, we  observe that while for the computable real numbers  there is no  computable numbering \cite{Ceitin,Lof}
  as well as for the computable points of a complete  computable metric space \cite{Brattka_baire},
  for a weakly effective $\omega$--continuous domain  there is a computable numbering of the computable elements  \cite{Weihrauch_cpo}.
Below we point out a natural sufficient condition on the family of basic neighborhoods of computable elements that guarantees the existence of a principal
computable numbering. We show that weakly effective $\omega$--continuous domains satisfy this condition.

\begin{defi}\cite{Ershov_Num_2}\label{wn_family}
Let $S$ be  a family of computably enumerable subsets of $\omega$. $S$ is called a $wn$-family if
there exists a partial computable function $\sigma:\omega\to\omega$ such that \hfill
\begin{enumerate}
\item[(i)] if $\sigma(n)\downarrow$ then $W_{\sigma(n)}\in S$ and
\item[(ii)] if $W_n\in S$ then $n\in {\rm dom}(\sigma)$ and $W_n=W_{\sigma(n)}$.
\end{enumerate}
\end{defi}
From \cite{Ershov_Num_2} it follows that any $wn$-family $S_X$ has a standard principal computable numbering $\gamma:n\mapsto W_{\sigma(h_0(n))}$, where $h_0:\omega\to\omega$ is a total computable function such that
${\rm im} (h_0)={\rm dom} (\sigma)$.

\begin{thm}\label{numb_comp_elem}\cite{Korovina_mscs}
 If
$\SX$ is a $wn$-family then there exists  a principal computable (canonical) numbering $\bar{\gamma}:\omega\to X_c$.
\end{thm}
\begin{proof}

Let us define $\bar{\gamma}(n)=a\leftrightarrow A_a=\gamma(n)$.
\end{proof}
\begin{defi}
Let   $\bar{\gamma}:\omega\to \XC$ be a principal computable numbering and  $L\subseteq\XC$.
The set  $\Ix(L)=\{n|\bar{\gamma}( n)\in L\}$ is called  an {\em index set} for the subset $L$.
\end{defi}
\begin{prop}\cite{Korovina_mscs}\label{eo_ce}
Let  $\left( X,\tau,\alpha \right )$ be an \ee $T_0$-space and $\SX$ be a $wn$-family.
If  $K$ is effectively open in $X_c$ then $Ix(K)$ is computably enumerable.
\end{prop}
 In \cite{Korovina_CiE15} we already have shown that for an \ee topological space $X$ such that $\SX$ is a $wn$-family generalisations of Rice's theorem hold. Moreover from the results in \cite{Ershov_Num_2} it is easy to see the following. If $X_c$ has the least element then the principal computable numbering is complete. So several   results from classical numbering theory can  be generalised for this case, in particular, if $Ix(K)$ is computably enumerable then it is creative.
 So
at  first glance this class looks promising to generalise the Rice-Shapiro theorem. However  in Section~\ref{sec_Counterexamples} we construct a counterexample that shows the existence of
an \ee topological space $X$ such that $\SX$ is a $wn$-family but the Rice-Shapiro theorem does not hold.
This forces us to search  for stronger requirements on \ee topological spaces which, on the one hand, do not  restrict the class too much, on the other hand, guarantee that  the Rice-Shapiro theorem holds.

\section{ The Rice-Shapiro Theorem}\label{sec_Rice_Shapiro}
In this section we  recall the classical Rice-Shapiro theorem then introduce  the new class of modular $T_0$-spaces and   prove the generalised  Rice-Shapiro Theorem for this class.

\begin{thm}[Classical Rice-Shapiro]\cite{Rogers}\label{classical_RSH}
Let $K$ be  a class of c.e. sets.
Then $Ix(K)$ is computably enumerable if and only if there exists a strongly computable sequence $\{F_n\}_{n\in\omega}$ of finite subsets of $\omega$ such that $K=\{E \mbox{ is  a c.e. set}\mid (\exists{n\in\omega})\, E\supseteq F_n\}$.
In  modern topological terms, $Ix(K)$ is computably enumerable if and only if $K$ is effectively open in the space of c.e. subsets of $\omega$, where $\mathcal{P}(\omega)$ is endowed with  the Scott topology.
\end{thm}
We use the specialisation order $x\leq y\rightleftharpoons \mbox{ for all open $\mathcal{O}$  if $x\in \mathcal{O}$ then $ y\in \mathcal{O} $}$. For $B\subseteq X$ we use the notation $x\leq B$ if, for all $y\in B$, $x\leq y$.
\begin{defi}\label{d_modular_space}
An \ee $T_0$-space  $\left( X,\tau,\alpha \right )$ is called  a modular $T_0$--space if it  satisfies the following requirements:
\begin{enumerate}[label=Req\ \arabic*:]
\item[Req 1:] $\SX$ is a $wn$-family.
\item[Req 2:] There exist a computable sequence  $\{b_n\}_{n\in\omega}$ of computable elements and  a computable sequence $\{\mathcal{O}_n\}_{n\in\omega}$ of effectively open sets such that
    \begin{enumerate}
    \item $b_n\leq \mathcal{O}_n$, where $\leq$  is the specialisation order and
    \item for all $m\in\omega $ $\alpha(m)=\bigcup_{b_i\in\alpha(m)}\mathcal{O}_i$.
    \end{enumerate}
\end{enumerate}
\end{defi}

\begin{lem}\label{l_modular_prop}
Assume  that $X$ is a modular $T_0$-space. Then for any finite $V\subset \omega$, the following equality holds
\begin{align*}
&\bigcap_{i\in V}\alpha(i)=
\!\!\!\!\!\!\bigcup_{ b_j\in \!\!\bigcap\limits_{i\in V}\!\!\alpha(i)}\!\!\!\!\!\!\mathcal{O}_j.
\end{align*}
\end{lem}
\begin{proof}
If $V=\emptyset$ then the both sides of the equality are $X$.
Now we assume that $V\neq\emptyset$.

\noindent $\subseteq$). Let $x\in \bigcap_{i\in V}\alpha(i)$. By the definition of an \ee topological space, there exists a  c.e. set $E\subseteq \omega$ such that $\bigcap_{i\in V}\alpha(i)= \bigcup_{i\in E}\alpha(i)$. Suppose $x\in \alpha(k)$, $k\in E$. By Req~2, there exists $l\in\omega$ such that $b_l\in\alpha(k)$, $x\in \mathcal{O}_l$ and $\mathcal{O}_l\subseteq \alpha(k)$. Therefore
$b_l\in \bigcap_{i\in V}\alpha(i)$ and $x\in \mathcal{O}_l$.

\noindent $\supseteq$).
Assume $x\in \mathcal{O}_j$ and $b_j\in\bigcap_{i\in V}\alpha(i)$ for some fixed $j\in\omega$.
By Req~2, $b_j\leq x$. From the definition of the specialisation order it follows that $x\in \bigcap_{i\in V}\alpha(i)$.
\end{proof}
Below we show  that every weakly effective $\omega$-continuous domain is a modular $T_0$-space. In Section~\ref{sec_Counterexamples} we will see that they are a proper subclass of the modular $T_0$-spaces.
\begin{prop}\label{wn_family_wef_cont_dom}
Let $(D;B \leq,\bot) $ be a weakly effective $\omega$--continuous domain.
Then  $D$ endowed with the Scott topology is a modular $T_0$-space.
 \end{prop}
 \begin{proof}
 Without loss of generality assume $\beta:\omega\to B$ such that $\beta(0)=\bot$.
 It is worth noting that the Scott topology on $D$ is formed by the open sets
 $\mathcal{U}_n=\{x\mid\beta(n)\ll x\}$. We define the numbering of the~base of the~Scott topology as follows.
 \begin{align*}
 & n=0\rightarrow \alpha(n)=\emptyset;\\
 & n>0\rightarrow \alpha(n)=\mathcal{U}_{n}.
 \end{align*}

 In \cite{Korovina_CCA08} we have already proven that $(D,\alpha)$ is an \ee space. Now we show that this space satisfies all requirements  of Definition~\ref{d_modular_space}. Although in \cite{Korovina_mscs} we established that $\widehat{D_c}$ is a $wn$-family, the proof was not enough seeable and complete. Below we propose a complete proof of this property.

\noindent  Req 1:
 From the definition of $\alpha$ it follows that $\widehat{D_c}=\{\{n\mid x\gg \beta(n)\}\mid x\in D_c\}$.
 First let us note that
 there exists a strongly computable sequence $\{A_s\}_{s\in\omega}$ of finite subsets of $\omega^2$ such that
 \begin{enumerate}
 \item $\{A_s\}_{n\in\omega}$ is a presentation of  the c.e. set $\{(i,j)\mid \beta(i)\ll\beta(j)\} $;
 \item $A_s\subseteq \{0,\dots,s\}^2$;
 \item $A_s$ is a transitive relation on $\omega^2$.
  \end{enumerate}
  For that it is sufficient to take the transitive closure of any presentation of $\{(i,j)\mid \beta(i)\ll\beta(j)\} $.
Below  for $A_t(i,j)$ we use the informal  notation $\beta(i)\ll^t\beta(j)$ in order to emphasise that  $A_t(i,j)$ finitely approximates $(\{(i,j)\mid \beta(i)\ll\beta(j)\} $.
We will write    $\beta(i)\not\ll^t\beta(j)$ if $(i,j)\not\in A_t$.

 For a c.e. set $W_e$,  we simultaneously by stages construct the computable functions $h_e:\omega\to\omega\cup \{-1\}$ and $g_e:\omega\to\omega$ as follows.

 \noindent $\bf{Stage\, 0}$. $g_e(0)=0$ and $h_e(0)=0$.

 \noindent $\bf{Stage\, s+1}$.
Assume that $g_e(s)$ and $h_e(s)$ are already constructed.
Put
 \begin{align*}
h_e(s+1)=\left \{
\begin{array}{@{~}l@{~}l}
k & \mbox{ if }  k=\min\{ n \mid n\in  W^{s+1}_e \mbox{ and } \beta(n)\not\ll^{s+1}\beta(g_e(s))\}\\
-1 &\mbox{ if }  \mbox{there is no such $k$}
\end{array}
\right.
\end{align*}
and
\begin{align*}
g_e(s+1)=\left \{
\begin{array}{@{~}l@{~}l}
b&\mbox{if } b=\min\{ x >0\mid x\in  W^{s+1}_e, \,
\beta(x)\gg^{s+1}\beta(g_e(s)),\, \beta(x)\gg^{s+1}\beta(k)\}\\
&\,\,\,\, \mbox{ and } h_e(s+1)=k\geq 0\\
g_e(s)&\mbox{if there is no such $b$}.
\end{array}
\right.
\end{align*}

\noindent By construction, either $\beta(g_e(s+1))\gg^{s+1} \beta(g_e(s))$ or $g_e(s+1)= g_e(s)$. It is obvious that $g_e$ is computable uniformly in $e$.
Put
\begin{align*}
\alpha_c(e)=\sup_{s\in\omega} g_e(s).
\end{align*}
We define a sequence $\{W_{\sigma(i)}\}_{i\in\omega}$ of c.e. sets by the rule $W_{\sigma(e)}=\{n\mid \beta(n)\ll\alpha_c(e)\}$.

Since, by definition, $\{n\mid \beta(n)\ll\alpha_c(e)\}=\{n\mid (\exists s) \beta(n)\ll\beta(g_e(s))\}$ the function $\sigma$ can be chosen to be  computable.
In order to show that $\sigma$ is a required function it is sufficient to prove  the following properties.
\begin{enumerate}[label=Pr \arabic*:]
\item[Pr 1:] The function $\alpha_c$ is a  computable numbering of a subset of $D_c$, i.e.,

$(\forall e\in \omega) \, W_{\sigma(e)}\in \widehat{D_c}$,
\item[Pr 2:] If $a\in D_c$ and $W_e=\{n\mid \beta(n)\ll a\}$ then $W_e=W_{\sigma(e)}$, i.e.,  $\alpha_c(e)=a$. Therefore ${\rm im}(\alpha_c)=D_c$.
\end{enumerate}
The property Pr~1 follows from the  computability of $g_e$ and the equivalence
 $ n\in W_{\sigma(e)}\leftrightarrow  (\exists s)\, \beta(n)\ll\beta(g_e(s))$.
Now we show the  property Pr~2.
By construction, $\beta(g_e(s))\in W^s_e$, therefore $\beta(g_e(s))\ll a$. So $\alpha_c(e)\leq a$.
In order to prove $\alpha_c(e)\geq a$ we show that if, for some $n\in\omega$,  $\beta(n)\ll a$, i.e., $n\in W_e$  then
there exists $s\in\omega$ such that $\beta(g_e(s+1))\gg\beta(n)$.
Assume the contrary.   We choose the minimal $n\in W_e$ such that $n>0$ and, for all $s\in\omega$, $\beta(g_e(s+1))\not\gg\beta(n)$.
Let us consider the step $s_0$ with property:  if   $k\in W_e$ and $k\leq n$  then $k\in W^{s_0}_e$ and $(\forall k<n)\, \beta(k)\ll\beta(g_e(s_0+1))$.
Now we show that there exists $s_1>s_0$ such that $(\forall m\geq s_1)\, \beta(k)\ll^m\beta(g_e(m))$ for all $k<n$.

Indeed, there exists
 $t>s_0+1$ such that $\beta(k)\ll^t\beta(g_e(s_0+1))$ since $\beta(k)\ll\beta(g_e(s_0+1))$ for all $k<n$.
 In the sequence  $
 \beta(g_e(s_0+1)),\dots,  \beta(g_e(t))
$ of basic elements, either $g_e(i)=g_e(i+1)$ or  $\beta(g_e(i))\ll^{i+1}\beta(g_e(i+1))$ for $s_0+1\leq i<t$.
By the monotonicity of $\{A_s\}_{s\in\omega}$ and the transitivity of $A_t$, we have $\beta(k)\ll^t\beta(g_e(t))$.
Hence $s_1$ can be chosen to be  the minimal $t>s_0+1$ such that $(\forall k<n) \beta(k)\ll^t\beta(g_e(s_0+1))$.
Therefore, for any $s\geq s_1$ and $k<n$,  $\beta(g_e(s))\gg^{s+1}\beta(k)$.  By the construction of $h_e$, we have $h_e(s+1)\neq k$ for any $k<n$.
So $h_e(s+1)=n$ since
$\beta(g_e)(s+1)\not\gg \beta(n)$.  By the construction of $g_e$, for any $s\geq s_1$, we have
$g_e(s+1)=g_e(s)$.
However, by the interpolation property (see Proposition~\ref{interpolation_prop}), from $a\gg\beta(n)$ and $a\gg \beta(g_e(s_1))$ it follows that  there exists $x\in\omega$ such that $a\gg\beta(x)\gg\beta(n)$ and $a\gg\beta(x)\gg\beta(g_e(s_1))$.
Then, for some $s_2>s_1$, we have $a\gg^{s_2}\beta(x)$, $\beta(x)\gg^{s_2}\beta(g_e(s_1))$, $\beta(x)\gg^{s_2}\beta(g_e(s_1))$ and $x\in W^{s_2}_e$.
So, at the stage $s_2+1$, by the construction of $g_e$, $\beta(g_e(s_2+1))\gg^{s_2+1}\beta(n)$.
This contradicts to the assumption. Therefore $\alpha_c(e)\geq a$ and $\alpha_c(e)= a$.
Hence ${\rm im}(\alpha_c)=D_c$ and
 $\widehat{D_c}$ is a $wn$-family.

 \noindent  Req 2:
 Put $b_n=\beta(n)$ and $\mathcal{O}_n=U_n$ for $n>0$.
 By construction, $b_n\leq \mathcal{O}_n$.  Now we show  that
 $\mathcal{O}_n=\bigcup_{b_i\in\mathcal{O}_n}\mathcal{O}_i$.
In order to do that we use the interpolation  property, i.e., if $b_n\ll x$ then there exists $b_i$ such that $b_n\ll b_i\ll x$. Assume $x\in\mathcal{O}_n$. Then, by the  interpolation property, for some $i\in\omega$, $b_i\in\mathcal{O}_n$ and $b_i\ll x$. Hence $b_i\in \mathcal{O}_n$ and $ x\in\mathcal{O}_i$, i.e.,
$x\in \bigcup_{b_i\in\mathcal{O}_n}\mathcal{O}_i$. The inclusion  $\mathcal{O}_n\supseteq\bigcup_{b_i\in\mathcal{O}_n}\mathcal{O}_i$ follows from the monotonicity of the relation $\ll$, i.e., from $b_n\ll b_i\ll x$ it follows that $b\ll x$, so $x\in \mathcal{O}_n$.
 \end{proof}

\begin{thm}[Generalised Rice-Shapiro]\label{RSH}
Let $(X,\tau,\alpha)$ be  a modular $T_0$--space and $K\subseteq \XC$.
Then $Ix(K)$ is computably enumerable if and only if $K$ is effectively open in $X_c$. Moreover
the representation $K=\bigcup_{n\in W_i}\alpha(n)\bigcap K_c$ can be computed by the index of the c.e. set $Ix(K)$.
\end{thm}
The claim is based on Proposition~\ref{eo_ce} and  the following propositions and lemmas. We use notations from Definition~\ref{wn_family} and Theorem~\ref{numb_comp_elem}.

\begin{lem}[Branching lemma]\cite{Berger_93,ST94}\label{berger}
Let $V$ and $W$ be computably enumerable sets such that $W$ contains all computably enumerable indices of $V$. Let $\{V_p\}_{p\in\omega}$ be a presentation of $V$ and $r:\omega\to\omega$ be a total computable function. Then there are $e\in W$ and $p\in\omega$ such that $W_e=V_p\cup W_{r(p)}$.
Furthermore such $e$ and $p$ are computed uniformly from a computably enumerable index of $W$ and computable indices of the functions $\lambda p.V_p$ and $r$.
\end{lem}
 It is worth noting that the original Branching lemma has been proven in \cite{Berger_93} for index sets of partial computable functions. For our purposes it is more convenient to use the modified version of the Branching lemma for  index sets of c.e. sets from \cite{ST94}.

\begin{lem}\label{W_monot}
 For $K\subseteq X_c$ let us define  $W=\{n|n\in {\rm dom}(\sigma) \mbox{ and } W_{\sigma(n)}\in \SK\}$.
If $W$  is computably enumerable then $\SK$ is monotone, i.e., if $A\subseteq B$, $A\in \SK$ and $B\in \SX$ then $B\in \SK$.
\end{lem}
\begin{proof}
Let $A=W_{\sigma(m)}$ and $B=W_{\sigma(n)}$. Define $V=A$ and $r(p)\equiv \sigma(n)$.
It is worth noting that,  for any $i\in\omega$, if $W_l=W_{\sigma(m)}\in \SK$ then $W_{\sigma(l)}=W_l\in \SK$, so that $l\in W$. Therefore $W$ contains all computably enumerable indices of $V$. Hence  we can use Lemma~\ref{berger} to find
 $e\in W$ and $p\in\omega$ such that $W_e=V_p\cup W_{r(p)}=W^p_{\sigma(m)}\cup W_{\sigma(n)}=W_{\sigma(n)}$. Since $B\in \SX$ and $e\in W$, $W_{\sigma(e)}=W_e\in \SK$. Therefore  $B\in \SK$.
\end{proof}
\begin{prop}\label{monoton}
For $K\subseteq \XC$, if $Ix(K)$ is c.e. then from $a\in K$ and $a\leq b$ it follows that $b\in K$.
\end{prop}
\begin{proof}
The claim follows from Lemma~\ref{W_monot} and the property that if $a\leq b$ then $A_a\subseteq A_b$.
\end{proof}
\begin{prop}\label{ce_openK} Let $K\subseteq X_c$. If $\Ix(K)$ is c.e. then $K$ is open in $X_c$.
\end{prop}
\begin{proof}
Let us fix the sequences $\{\mathcal{O}_n\}_{n\in\omega}$ and $\{b_n\}_{n\in\omega}$ for $X$ from Definition~\ref{d_modular_space}
 and denote $C=Ix(K)$.
Assume that $K$ is not open. Then there exists $a\in K$ such that, for all $\mathcal{O}_n$, $n\in\omega$, if
 $a\in \mathcal{O}_n$ then
$\mathcal{O}_n\cap \XC\not\subseteq K$. Let us fix  one of such $a$ and a presentation $\{V_s\}_{s\in\omega}$ of $A_a$.
 Then we define $U_m=\{k\mid\bar{\gamma}(k)\in\bigcap_{i\in V_m}\alpha(i)\}$. It is easy to see that
 $\{U_m\}_{m\in\omega}$ is a computable sequence of c.e. sets and every $\bar{\gamma}(U_m)$ contains $a$.
Our goal is to construct a  computable function $h:\omega\to\omega$ such that $h(m)\in U_m\setminus C$.
First we note that there exists $n=n(m)\in\omega$ such that
\begin{enumerate}
\item $a\in\mathcal{O}_n$ and
\item $b_n\in\bigcap_{i\in V_m}\alpha(i)$.
\end{enumerate}
Indeed,  by Lemma~\ref{l_modular_prop},
\begin{align*}
&a\in \bigcap_{i\in V_m}\alpha(i)=
\!\!\!\!\!\!\bigcup_{ b_j\in \!\!\bigcap\limits_{i\in V_m}\!\!\alpha(i)}\!\!\!\!\!\!\mathcal{O}_j.
\end{align*}

The existence of a required $n$ follows from the formula above.
 Since this search is effective, the function $n(m)$ is computable.

 It is worth noting that, by assumption, $\mathcal{O}_n\cap \XC\not\subseteq K$ and,  by Req 2, $b_n\leq \mathcal{O}_n$. From  Proposition~\ref{monoton} it follows that $b_n\not\in K$.
 At the same time, by construction, $\bar{\gamma}^{-1}(b_n)\subseteq U_m$.
 Since $\bar{\gamma}$ is a principal computable numbering there exists a computable function $f:\omega\to\omega$ such that $b_n=\bar{\gamma}(f(n))$.
 Hence we can put
$h(m)=f(n(m))$. By construction, $h(m)\in U_m\setminus C$.

Now we are ready to use Lemma~\ref{berger}.
We put $V=A_a$, $W=\{n|n\in {\rm dom}(\sigma) \mbox{ and } W_{\sigma(n)}\in \SK\}$ and the computable function $r:\omega\to\omega$ satisfying $W_{r(p)}=\gamma(h(p))$.
Since $a$ is a~computable element it follows that  for any $n\in\omega$ if $W_n=V$ then $W_{\sigma(n)}=W_n\in \SK$, so  $n\in W$. Therefore $W$ contains all computably enumerable indices of $V$. We can use Lemma~\ref{berger} to find
$e\in W$  and $p\in\omega$ such that
\begin{eqnarray*}
W_e=V_p\cup \gamma(h(p)),
\end{eqnarray*}
where $V_p$ is defined above.  If $i\in V_p$ then $\bar{\gamma}(h(p))\in\alpha(i)$, i.e., $i\in \gamma(h(p))=A_{\bar{\gamma}(h(p))}$. Therefore  $V_p\subseteq \gamma(h(p))$ and $W_e= \gamma(h(p))$. On the one hand $W_e=W_{\sigma(e)}\in \SK$ since $W_e\in \SX$ and $e\in W$ on the other hand
$\gamma(h(p))\not\in \SK$ by  the construction of $h$. We get a  contradiction.
\end{proof}

\begin{proof}(Theorem~\ref{RSH})

\noindent $\rightarrow).$   By Proposition~\ref{ce_openK} and Req 2, $K=\bigcup_{b_n\in K} \mathcal{O}_n\cap\XC$. Since $\Ix(K)$ is c.e. and $\{b_n\}_{n\in\omega}$ is a computable sequence, the set $\{n\mid b_n\in K\}$ is c.e. Therefore $K$ is effectively open in $X_c$.

\noindent $\leftarrow).$ The claim follows from Proposition~\ref{eo_ce}.
\end{proof}

\section{Several Constructions of \ee $T_0$-spaces}\label{sec_constructions}
In this section we show a few general approaches  for constructing \ee $T_0$-spaces with particular properties.

\subsection{From  a $wn$-family to an \ee $T_0$-space}\label{subsec_from_wn_to_top_space}
In this subsection we show how to construct from any   $wn$-family $(S,\gamma)$  an \ee $T_0$-space which inherits  some of the properties of $S$.

Assume that $S\subseteq \mathcal{P}(\omega)$ is a $wn$-family  and $\gamma:\omega\to S$ is its principal computable numbering.
 %
 First, it is worth noting that $S$ can be considered as a subspace of $\mathcal{P}(\omega)$ with the Scott topology.  Let $X=(S,\tau,\beta)$, where $\beta(n)=\{V\in S\mid D_n\subseteq V\}$. In order to show that $X$ is an \ee $T_0$-space $X$ we consider its homeomorphic copy $(X_S,\tau_S,\alpha_S)$ defined as follows.
 \begin{eqnarray*}
&&X_S=\omega\slash\sim, \mbox{ where } i\sim j\leftrightarrow \gamma(i)=\gamma(j);\\
&& \, [i]_\sim\in \alpha_S(n) \leftrightarrow D_n\subseteq \gamma(i).
 \end{eqnarray*}
\begin{prop}\hfil\label{constr_S_to_X}
\begin{enumerate}
\item $(X_S,\tau_S,\alpha_S)$ is an \ee $T_0$--space such that every element of $X_S$ is computable and $\widehat{X_S}$ is a $wn$-family.
\item The topological spaces $X$ and $X_S$ are homeomorphic.
\item The space $X$ is an \ee $T_0$--space such that  every element of $X$ is computable and $\widehat{X}$ is a $wn$-family.
\item Let  $K\subseteq S$. The transfer from $(S,\gamma)$ to $X$ preserves effective openness of $K$ and its index sets.

\end{enumerate}
\end{prop}
\begin{proof}
\noindent $(1)$ It is clear that $\alpha_S(n)\cap \alpha_S(n)=\alpha_S(l)$, where $D_l=D_n\cup D_m$, and the set $\{i\mid\alpha_S(i)\neq\emptyset\}$ is computably enumerable.
Therefore  $(X_S,\tau_S,\alpha_S)$ is an \ee $T_0$--space.
All elements of $X_S$ are computable since $\{n\mid i\in\alpha_S(n)\}=\{n\mid D_n\subseteq \gamma(i)\}$.
Let us  define $\gamma^\ast(i)=\{n\mid D_n\subseteq \gamma(i)\}$.
Now we show that $\widehat{X_S}$ is a $wn$-family.
It is worth noting that $S$  consists of   all $W_{\sigma(n)}$ and if $W_n\in S$ then $W_{\sigma(n)}=W_n$.
It is easy to see that there exists a computable function $g:\omega\to \omega$ such that $W_{g(m)}=\bigcup_{n\in W_m}D_n$. Define $\sigma^\ast:\omega\to\omega$ as follows:
\begin{eqnarray*}
W_{\sigma^\ast(m)}=\{n\mid D_n\subseteq W_{\sigma(g(m))}\}.
\end{eqnarray*}
We check the following properties:
\begin{enumerate}[label=(P \arabic*)]
\item[(P1)] If $\sigma^\ast(m)\downarrow$ and $W_{\sigma(g(m))}=\gamma(i)$ then $W_{\sigma^\ast(m)}=\gamma^\ast(i)$, i.e.,  $W_{\sigma^\ast(m)}\in \widehat{X_S}$.
    \item[(P2)] If $W_m\in\widehat{ X_S}$ then $W_{\sigma^\ast(m)}=W_m$.
\end{enumerate}
In order to show (P1) assume that $\sigma^\ast(m)\downarrow$, so  $\sigma(g(m))\downarrow$.
Then $W_{\sigma(g(m))}=\gamma(i)\in S$ for some $i\in\omega$. Therefore $W_{\sigma^\ast(m)}=\gamma^\ast(i)\in \widehat{X_S}$. In order to show (P2) assume that $W_m\in \widehat{X_S}$.  By definition it means that $W_m=\{n\mid D_n\subseteq \gamma(i)\}$ for some $i\in\omega$. Therefore $W_{g(m)}=\gamma(i)$ and $W_{\sigma(g(m))}=W_{g(m)}$. As a corollary,
\begin{eqnarray*}
W_{\sigma^\ast(m)}=\{n\mid D_n\subseteq \gamma(i)\}=\gamma^\ast(i)=W_m.
\end{eqnarray*}
Hence $\widehat{X_S}$ is a $wn$-family.

\noindent $(2)$ Let us define $f:\omega\slash\sim\to S$ as $f([i]_\sim)=\gamma(i)$ and $\beta(i)=f(\alpha_S(i))$. It is clear that $f$ is a homeomorphism between $([i]_\sim),\tau_S,\alpha_S)$ and  $(S,\tau,\beta)$.

\noindent The claims (3) and (4) are straightforward by  the construction and the previous  claims.
 \end{proof}
\subsection{From a tree $T$ to a modular $T_0$-space}\label{subsec_from_tree_to_top_space}
In this subsection we show how to generate a modular $T_0$--space  $X_T$  from any computable tree $T$ without computable infinite paths.
We use the standard notations  $\omega^{<\omega}$  and  $\omega^\omega$ and, by default, we endow the set $\omega^\omega$ with the standard order  $ x\sqsubseteq y \equiv (\forall i\in\omega)  \mbox{ if }x(i)\downarrow \mbox{ then } y(i)=x(i)$ (c.f. \cite{Rogers}).
We take a standard  agreement  that a  downward  closed nonempty subset $T\subseteq {\omega}^{<\omega}$ is a tree.
Below we also use the lexicographic and Kleene-Brouwer  orders   on $\omega^{<\omega}$ defined as follows.
\begin{eqnarray*}
&& x\preceq y \equiv x\sqsubseteq y \vee (\exists i\in \omega)(\forall j< i) x(j)=y(j) \wedge x(i)<y(i);\\
&& x\leq_{KB} y \equiv x \sqsupseteq y \vee (\exists i\in \omega)(\forall j< i)(x(j)=y(j)\wedge x(i)<y(i)).
\end{eqnarray*}
A set $p\subseteq T$ is called a {\em partial path} if $p$ is downward  closed and linear ordered by $\sqsubseteq$. A maximal partial path is called  a {\em path}.
Let $[T]_{p}$ denote all partial paths and  $[T]^{fin}_{p}$ denote all finite partial paths.
It is easy to see that there is a straightforward bijective
correspondence between the  vertices of a tree and its finite
nonempty partial paths. Indeed, if 
$x\in T$ then the corresponding path is $p_x=\{y\mid y\sqsubseteq x\}$. Also every finite partial path is equal to $p_x$ for an appropriate $x\in T$. It is worth noting that  infinite paths correspond to elements of $\omega^\omega$.
 In a natural way we define the order $\sqsubseteq$ for partial paths as follows.
\begin{eqnarray*}
&& p \sqsubseteq q \equiv p \subseteq  q.
\end{eqnarray*}
 It is easy to see that if $x\sqsubseteq  y$ then  $p_x \sqsubseteq p_y$. Below we use the notation  $x \sqsubseteq p$ if   $p_x \sqsubseteq p$.

Now we assume that  $T$ is  a computable (recursive) \cite{Rogers} tree  without computable infinite paths and   $ \delta:\omega\to \omega^{<\omega}$ is a canonical bijective computable  numbering.  Put $S_T=\delta^{-1}([T]^{fin}_{p})$.
In order to show that $S_T$ is a $wn$-family
 we  construct $W_{\sigma(n)}$ by stages.

 \noindent $\bf{Stage\, 0}$. $W^{ 0}_{\sigma(n)}=\emptyset$.

  \noindent $\bf{Stage\, s+1}$.
  First,  we define $B\subseteq W^s_n$ with the following properties: if $b\in B$ then
  \begin{enumerate}
  \item $\delta(b)\in T$,
  \item  $(\forall c\in\omega)\, \delta(c)\preceq\delta(b)\rightarrow c\in W^s_n,$
 \item  there is no $\tilde{b}\in  W^s_n$ such that $\delta(b)\le_{KB}\delta( \tilde{b})$,  $\delta(\tilde{b})\not\preceq \delta(b)$ and $\tilde{b}\in W^{ s}_{\delta(n)}.$
  \end{enumerate}
 Then,  we choose ${ b}\in B$ such that $\delta({ b})$ is a $KB$-min element in $\delta(B)$  and define $$W^{ s+1}_{\sigma(n)}= W^{ s}_{\sigma(n)}\cup \{c\mid \delta(c)\preceq\delta({ b})\}.$$

   \noindent Put $W_{\sigma(n)}=\bigcup_{{ s}\in\omega}W^{ s}_{\sigma(n)}$ and $S=\{W_{\sigma(n)}\mid n\in\omega\}$. Let us point out properties of $S$.
   \begin{enumerate}
   \item Every $W_{\sigma(n)}$ is finite since there are no computable infinite paths in $T$. In other words,  $S$ contains only the pre-images of all finite paths in $T$ under $\delta$.
   \item If $W_n\in S$ then, by construction, we put all its elements in $W_{\sigma(n)}$. So $W_{\sigma(n)}=W_n$.
   \end{enumerate}
 Therefore   $S=S_T$ and as  a corollary   $S_T$ is a $wn$-family.

Define $X_T=[T]^{fin}_p$,  $\alpha(i)=\mathfrak{A}_{\delta(i)}$, where
$\mathfrak{A}_x=\{p\in[T]^{fin}_p\mid x\sqsubseteq p\}$.

\begin{prop}\label{X_T}
The space $(X_T,\tau,\alpha)$   is a modular $T_0$--space.
\end{prop}
\begin{proof}
By definition, $\mathfrak{A}_u\cap \mathfrak{A}_v= \mathfrak{A}_u$ if $u\sqsubseteq v$ and $\mathfrak{A}_u\cap \mathfrak{A}_v= \emptyset$ if $u$ and $ v$ are incomparable.
 The set $\{i\mid\alpha(i)\neq\emptyset\}$ is
 computably enumerable since $T$ is computable. So $X_T$ is an  \ee $T_0$--space, moreover
 every element of $X_T$ is computable. For  $p\in X_T$, $A_p=\delta^{-1}(p)$ hence
 $\widehat{X_T}=S_T$. Therefore  $\widehat{X_T}$ is a $wn$-family.
 Since $T$ is computable, $\delta^{-1}(T)$ is computable. Let $\delta^{-1}(T)=\{c_1<c_2<\dots<c_n\dots\}$.
 Now  we define $b_n=\{x\in T\mid x\sqsubseteq \delta(c_n)\}$ and $\mathcal{O}_n=\mathfrak{A}_{\delta(c_n)}$. It is clear that $X_T$ satisfies Req 1 and Req 2 of Definition~\ref{d_modular_space}.
 It is worth noting that the order $\sqsubseteq$ coincides with the specialisation order $\leq$.
\end{proof}
Let us note that if in a  standard way one     identifies vertices and finite partial paths of $T$ then $X_T=T$ with the topology formed by $\mathcal{U}_x=\{a\in T\mid x\sqsubseteq a\}$.

\section{Counterexamples}\label{sec_Counterexamples}
 Using the techniques from Section~\ref{sec_constructions} we construct
a few examples of \ee $T_0$-spaces with particular properties.

\subsection{Without the Rice-Shapiro theorem}\label{subsec_without_Rice_Shapiro}
In this subsection we provide an \ee $T_0$-space $X$ such that $\SX$ is a $wn$-family however
 the Rice-Shapiro theorem does not hold for the computable elements.

\begin{defi}\cite{Ershov_Num_73_1,Ershov_Num_75_2}
We say that $S\subseteq {\mathcal P}(\omega)$ is effectively discrete   if there exists a strongly computable family $\{F_n\}_{n\in\omega}$ of finite subsets of $\omega$ such that
\begin{eqnarray*}
(\forall A\in S) (\exists\,  n\in\omega) A\supseteq F_n\wedge (\forall A\in S)(\forall B\in S)
( A\supseteq F_n\wedge  B\supseteq F_n \rightarrow A=B).
\end{eqnarray*}
\end{defi}

It is worth noting that $K\subseteq S$ is effectively open in $S$ considered as a subspace of ${\mathcal P}(\omega)$
with the Scott topology  if and only if  there exists a strongly computable family $\{F_n\}_{n\in\omega}$ of finite subsets of $\omega$ such that
$
(\forall A\in S) \left( A\in K \leftrightarrow (\exists n\in\omega) A\supseteq F_n\right).
$
\begin{defi}\cite{Ershov_Num_2}
Let $S$ be  a set of computably enumerable subsets of $\omega$. A  computable numbering $\gamma:\omega\to S$ is called  a positive computable numbering if
 $\{(n,m)\mid \gamma(n)=\gamma(m)\}$ is computably enumerable.
 \end{defi}
\begin{prop}\cite{Vjugin}
There exists a $wn$-family $S\subseteq \mathcal{P}(\omega)$  with a   positive  principal computable numbering that is not effectively discrete.
\end{prop}
Let us fix a $wn$-family $S\subseteq \mathcal{P}(\omega)$   that is not effectively discrete and its positive  principal computable numbering $\gamma:\omega\to S$. Let $c:\omega\times\omega\to\omega$ be  the Cantor pairing function.
Define $S^\ast=\{c(A\times B)\mid A,\, B\in S\}$.
\begin{lem}\label{S_ast_wn-family}
The set $S^\ast$ is a $wn$-family.
\end{lem}
\begin{proof}

Let  $\sigma:\omega\to\omega$ be a partial computable function for the $wn$-family $S$ that satisfies the conditions of Definition~\ref{wn_family}.
The  partial computable function $\sigma^\ast:\omega\to\omega$ is defined by the equation
$W_{\sigma^\ast}(n)=c(W_{\sigma(a(n))}\times W_{\sigma(b(n))})$, where
$W_{a(n)}=\{x\mid \exists y\, c(x,y)\in W_n\}$ and
$W_{b(n)}=\{x\mid \exists x\, c(x,y)\in W_n\}$.
It is clear that  $\sigma^\ast$ satisfies the conditions of Definition~\ref{wn_family}. Hence
$S^\ast$ is a $wn$-family.
\end{proof}
\begin{prop}\label{ce_not_eff_open}
Let $K=\{c(A\times A)\mid A\in S\}$. Then $Ix(K)$ is c.e. but $K$ is not effectively open in $S^\ast$.
\end{prop}
\begin{proof}
Since $\gamma$ is a positive computable numbering, $Ix(K)$ is computably enumerable.
Assume that $K$ is  effectively open in $S^\ast$, i.e.,  for an appropriate
c.e. $I\subseteq \omega$, $K=\{C\in S^\ast\mid (\exists i\in I)\, C\supseteq D_i\}$, where
$ D_i$ is  a  finite set. Therefore, by the definition of $K$, there exists
a computable function $h:\omega\to\omega$ such that $D_{h(i)}=\{x\mid (\exists y\in\omega) (c(x,y)\in D_i \vee c(y,x)\in D_i) \}$.
  So, $C\supseteq c(D_{h(i)}\times D_{h(i)})\supseteq D_i$. Then, for $A,\, B\in S$, $A=B\leftrightarrow (\exists i\in\omega) A,\, B \supseteq D_{h(i)}.$ Therefore, $S$ is effectively discrete. We get  a contradiction.
\end{proof}

\begin{thm}
There exists an \ee $T_0$-space $(X,\tau,\alpha)$ such that $\SX$ is a $wn$-family but  the Rice-Shapiro theorem does not hold for the computable elements. \end{thm}
\begin{proof}
The claim follows from Proposition~\ref{constr_S_to_X} and Proposition~\ref{ce_not_eff_open}.
\end{proof}

\subsection{Non-dcpo}\label{subsec_non_dcpo}
In this subsection we show that the class of modular $T_0$-spaces is wider than the  weakly effective $\omega$--continuous domains.

\begin{thm} There exists a modular $T_0$-space $(X,\tau,\alpha)$ such that $(X,\leq)$ is  not a dcpo, where $\leq$ is the specialisation order.

\end{thm}
\begin{proof}
Let $T$ be  a computable tree without infinite computable paths and  with at least one infinite non-computable path. Consider $X_T$ from the Section~\ref{subsec_from_tree_to_top_space}.
Since $T$ has an infinite non-computable path, $(X_T,\leq)$ is not a  dcpo.
\end{proof}

\section*{Acknowledgement}
 We wish to acknowledge fruitful discussions with Dieter Spreen on effective topological spaces.

\end{document}